\documentclass[aps,prx,10pt,twocolumn,superscriptaddress,nofootinbib,longbibliography]{revtex4-2}
\newif\iflong\longtrue
\newif\ifstateonly\stateonlytrue

\usepackage[T1]{fontenc}
\usepackage{lmodern}
\usepackage{amsmath,amssymb,amsthm,mathtools}
\usepackage{microtype}
\usepackage{graphicx,booktabs}
\usepackage{xcolor}
\usepackage[normalem]{ulem}
\usepackage{pgfplots}
\pgfplotsset{compat=1.18}
\usetikzlibrary{arrows.meta,positioning}
\usepackage[hidelinks]{hyperref}
\hypersetup{pdftitle={Complexity Barriers to State Preparation in Quantum Approximate Optimization},pdfauthor={Stuart Hadfield},pdfsubject={Gain-calibrated state-preparation hardness for quantum approximate optimization}}
\newtheorem{theorem}{Theorem}
\newtheoremstyle{informal}{2pt}{2pt}{\normalfont}{}{\bfseries}{.}{0.5em}{}
\theoremstyle{informal}
\newtheorem*{informaltheorem}{Theorem (informal)}
\theoremstyle{plain}
\newtheorem{corollary}{Corollary}
\newtheorem{proposition}{Proposition}
\newtheorem{revisedproposition}[proposition]{\revcut{Corollary}\revnext{Proposition}}
\newtheorem{definition}{Definition}
\DeclareMathOperator{\Tr}{Tr}
\newcommand{\NP}{\mathrm{NP}}
\newcommand{\BQP}{\mathrm{BQP}}
\newcommand{\Id}{I}
\newcommand{\Cstar}{C^\star}
\newcommand{\gstar}{g^\star}
\newcommand{\genc}[1]{g_{\mathrm{enc}}(#1)}
\newcommand{\gdec}[1]{g_{\mathrm{dec}}(#1)}
\newcommand{\gqstar}{g_{\mathrm q}^{\star}}
\newcommand{\Renergy}[1]{R_{\mathrm E}(#1)}
\newcommand{\Renc}[1]{R_{\mathrm{enc}}(#1)}
\newcommand{\Rdec}[1]{R_{\mathrm{dec}}(#1)}
\newcommand{\Rq}[1]{R_{\mathrm q}(#1)}
\newcommand{\E}{\mathbb E}

\newcommand{\poly}{\operatorname{poly}}
\newcommand{\revnew}[1]{#1}
\newcommand{\revlatest}[1]{#1}
\newcommand{\revnext}[1]{#1}
\newcommand{\revcut}[1]{}
\newcommand{\branchadd}[1]{#1}
\makeatletter
\def\branchcut@words#1 #2\@nil{%
  \sout{#1}%
  \if\relax\detokenize{#2}\relax
  \else\space\branchcut@words#2\@nil\fi}
\newcommand{\branchcut}[1]{}
\makeatother

\newcommand{\subsectionview}[2]{\iflong\subsection{#2}\label{sec:#1}\else\par\emph{#2.---}\fi}
\newcommand{\techheading}[2]{\iflong\section{#2}\label{app:#1}\else\subsection*{#2}\fi}

\usepackage{capt-of}
\begin{document}
\raggedbottom

% ===== BEGIN shortest/frontmatter.tex =====
\newif\ifarxivtagline
\arxivtaglinetrue % PDF-only arXiv v1 tagline; disable for the journal PDF.
\title{Complexity Barriers to State Preparation in Quantum Approximate Optimization%
\ifarxivtagline\branchadd{\\[0.5ex]{\normalfont\small\itshape No pain, no gain}}\fi}
\author{Stuart Hadfield}
\email{shadfield@usra.edu}
\thanks{ORCID: \href{https://orcid.org/0000-0002-4607-3921}{0000-0002-4607-3921}.}
% \affiliation{Quantum Artificial Intelligence Laboratory, NASA Ames Research Center, Moffett Field, California 94035, USA}
\affiliation{USRA Research Institute for Advanced Computer Science, Mountain View, California 94043, USA}
\date{\today}
\begin{abstract}
For many important optimization problems we are restricted to approximate solutions in practice due to computational complexity.
Distinct from the exact optimization setting, approximate optimization admits performance measures beyond whether the optimum is found, with different tradeoffs and complexity.
For MaxCut, a near-unity (ordinary) approximation ratio can coexist with near-zero improvement (gain) over a random cut. For the standard encoding, the unconditional classical MaxCut-Gain hardness gap implies that \emph{any uniformly efficient quantum or hybrid procedure recovering a fixed positive fraction of the optimal classical gain on every input, with at least inverse-polynomial success probability, would place $\NP$ in $\BQP$}. Such a procedure is therefore believed impossible under standard assumptions.
We broadly address where our worst-case barriers do or do not apply across the quantum algorithm landscape. 
We prove that the barrier survives quantum random access optimization (QRAO) compression and applies throughout the mean-energy range $\Cstar\le E_\rho\le\lambda_{\max}(H_d)$ between the classical and relaxed optimal values. 
For every input, a product state attains the classical optimum. Thus the barrier to reaching the classical threshold does not arise from a need for entanglement.
For $d\in\{2,3\}$ variables per qubit, the known decoder transfers encoded energy gain to decoded mean gain by the exact factor $1/d^2$. Combining this identity with MaxCut-Gain hardness gives an operational preparation barrier for QRAO. We also construct hard $n$-qubit families with relative quantum relaxation excess $\Theta(1/n)$, while the maximally mixed state has energy approximation ratio $1-\Theta(1/n)$, zero encoded energy gain, and hence zero decoded mean gain.
Our results separate the effects of relaxation tightness and energy approximation ratio from operational accessibility, motivating more comprehensive practical accounting of decoded gain, readout, precision, and end-to-end cost in benchmarking and performance assessment.
\end{abstract}
\maketitle

% ===== END shortest/frontmatter.tex =====

\section{Introduction}\label{sec:intro}

% ===== BEGIN shortest/introduction.tex =====
Quantum optimization encompasses a broad set of approaches using quantum computers to obtain classical solutions to hard combinatorial problems. Different settings include exact, approximate, and heuristic regimes, whose complexity, \revnew{resource requirements}, performance, and benchmarking claims must be carefully distinguished~\cite{Abbas2024}. \revnew{For both} near-term and polynomially bounded settings, standard computational complexity assumptions %suggest 
imply that approximate solutions are the best we can hope to obtain efficiently for many important problems in the worst case. Existing performance guarantees, when obtainable, often concern the mean energy of prepared quantum states, which %, after a specified readout, 
bounds the expected objective value and hence the %ordinary 
algorithm approximation ratio. 
At the same time, concentration bounds can sometimes constrain the distribution of cuts produced by repeated runs on one specified problem instance. 
%This is distinct from averaging performance over a randomly drawn ensemble of instances. 
%At the same time, concentration bounds can limit the probability that repeated runs on a given problem instance produce solutions substantially better than the mean. 
Hence a common focus of quantum optimization
research has been to produce or demonstrate quantum
algorithms with improved approximation ratio bounds,
relative to previous ans\"atze or to classical benchmarks.

The achieved ordinary approximation ratio can nevertheless be an incomplete performance indicator, if not misleading. For example, a fixed approximation ratio does not indicate whether the same performance can be achieved by a classical algorithm, or how the algorithm performs on typical instances.  
This issue is reflected in the classical literature, where metrics beyond ordinary approximation ratio are also considered that exploit the additional degrees of freedom %opened by 
possible from relaxing the exact solution criterion\revnext{~\cite{DemangePaschos1999}}. Moreover, both %relative 
problem complexity and algorithm assessment can change when different metrics are considered. To our knowledge these alternative metrics have not yet been %broadly 
widely considered in the quantum optimization literature. 

Concretely, we %study 
elucidate this approximation complexity distinction for quantum approaches to MaxCut, a prototypical NP-hard optimization problem commonly considered for quantum algorithms. \revlatest{We instead assess in terms of the achievable performance gain above random assignments.} For MaxCut, its classical approximation theory in terms of gain is indeed distinct from that for the ordinary approximation ratio~\cite{HastadVenkatesh2004,KhotODonnell2009}. %\revlatest{The published classical hardness gap translates this distinction into a barrier to efficient state preparation.}
We leverage the classical hardness gap to 
% this distinction 
\revcut{translate into an achievable gain barrier to efficient quantum state preparation}\revnext{translate this distinction into a barrier to efficient quantum state preparation}.

We first consider the standard uncompressed (Ising spin) encoding, denoted by the case $d=1$, where each sample (computational basis measurement) directly returns a candidate cut. 
This encoding is used widely and across a variety of problems, from the quantum alternating operator ansatz (QAOA)~\cite{Farhi2014,Hadfield2019,HadfieldHamiltonians2021} to quantum annealing~\cite{Kadowaki1998,FarhiAdiabatic2000,lucas2014ising},
among many others. 
We then %consider 
%generalize 
extend to compressed encodings of quantum random access optimization (QRAO), where each qubit stores $d=2$ or $3$ classical variables~\cite{Fuller2024,Teramoto2023}. In QRAO, the input problem and compression choices determine the relaxed quantum Hamiltonian~$H_d$ to optimize. Decoding via local rounding transfers its centered (i.e., relative to the random guessing mean) energy, which we call the \textit{encoded energy gain}, to the \textit{decoded problem gain} by an exact multiplicative factor for arbitrary input states, including entangled states~\cite{Fuller2024}. We use this identity together with the classical gain gap to constrain uniform QRAO state preparation. %Consequently, a QRAO energy guarantee may impose a strong classical-gain target even when the corresponding ordinary approximation-ratio guarantee appears weak.
Consequently, a QRAO energy guarantee may impose a strong gain requirement or obstruction even when
the corresponding ordinary approximation ratio appears relatively weak.

For $d=1$, the classical MaxCut-Gain gap directly yields the BQP consequence stated below~\cite{Dinur2025,KhotMinzerSafra2023}. %The main 
A primary additional question is whether compression weakens this barrier. Previously, Teramoto et al.\ \revcut{identified as future work to characterize the hardness of finding a relaxed state whose energy exceeds the classical MaxCut value, together with some discussion of performance in terms of MaxCut-Gain}\revnext{identified characterizing the hardness of finding a relaxed state whose energy exceeds the classical MaxCut value, together with performance in terms of MaxCut-Gain, as future work}~\cite[Secs.~4.2.2--4.2.3]{Teramoto2023}. We answer this question for both the two-to-one and three-to-one QRAO encodings for every fixed target energy $E_\rho$ satisfying 
$E_\rho-E_0\ge\tau(\Cstar-E_0)$, 
with $\Cstar$ the optimal cut value, $E_0$ the expected cut value under uniform random guessing, and 
%$E_\rho-W/2\ge\tau(\Cstar-W/2)$ with 
any constant $0<\tau\le1$. 
%Above the %classical optimal threshold 
%optimal cut value target corresponding to $\tau=1$, 
\revcut{The barrier also applies to every possible compressed energy threshold $\lambda_{\max}(H_d)\geq E_\rho\geq\Cstar$ above the optimal cut value energy target.}\revnext{The barrier also applies throughout the compressed-energy range $\Cstar\le E_\rho\le\lambda_{\max}(H_d)$.} %corresponding to $\tau=1$. . 
%\revnew{The compressed proof composes the classical gap with the known decoder and an explicit mean-to-tail step. Its new content is the gain accounting, end-to-end preparation consequence, and near-tight family, not a new classical gap, QMA-hardness result, or quantum hardness-of-approximation technique.}
%
Our main theorem concerns what energy values can be efficiently obtained, 
and is complementary but distinct from recent work~\cite{HadfieldQRAOComplexity2026} that identified NP-, StoqMA-, and QMA-complete QRAO energy promise problems
concerning the relaxed optimum $\lambda_{\max}(H_d)$. A separate companion study of compression, observable magnitude margins, and readout cost~\cite{HadfieldNFC2026} uses information-theoretic bounds rather than the MaxCut-Gain gap.

\begin{informaltheorem}
%For the standard uncompressed Ising encoding,
\revnext{For the standard uncompressed encoding, }suppose there exists a uniformly efficient quantum or hybrid algorithm that for every MaxCut instance and with inverse-polynomial success probability prepares a state whose computational basis measurement recovers in expectation any fixed positive fraction of the optimal classical gain above random. This implies $\NP\subseteq\BQP$, and hence no such algorithm can exist under standard \revcut{complexity theoretic}\revnext{complexity-theoretic} assumptions. 
We prove the result holds for standard uncompressed Ising encodings, as well as for two-to-one and three-to-one QRAO encodings when the prepared state's encoded energy gain is measured against the optimal classical gain of the input graph. 
\end{informaltheorem}

This \revcut{state preparation}\revnext{state-preparation} barrier concerns broad classes of quantum algorithms including those based on quantum circuits, variational approaches, annealing, open-system, measurement-based, hybrid methods, and many others whose end-to-end implementations %---including control selection, preparation, readout, decoding, and repetitions---
are uniformly realizable to required accuracy by polynomial-size adaptive circuits (i.e., polynomial-time classical feedforward). 
\revnext{Furthermore the barrier is independent of hardware regime. It applies to noisy intermediate-scale quantum (NISQ), early fault-tolerant, and fully fault-tolerant approaches~\cite{Preskill2018, Sanders2020,HeEFTQC2024} whenever the complete implementation meets the same uniform-efficiency, accuracy, and success requirements.}
Here efficient and uniform implementation is critical and concerns all algorithm aspects including state preparation, parameter setting, control and precision, measurement readout, decoding and classical processing, success probability and required repetitions, and other factors. 
%\revcut{The result does not claim to cover all possible schemes, in particular leaving open effects of solution advice, nonuniform or exponentially specified controls, as well as guarantees over structured or typical instances, or claims re higher moments and tails of acheivable output distribtions.}
\revnext{The result does not cover every possible scheme. In particular, it leaves open the effects of supplied solution advice and nonuniform or exponentially specified controls, as well as guarantees restricted to structured or typical instances, or refined claims about tails and higher moments %or tails 
of achievable output distributions.}
\revlatest{Definition~\ref{def:operational} below gives the precise scope, and Theorem~\ref{thm:prep} gives the formal statement.}

%Compression also makes a structural distinction visible. The classical threshold \(E_\rho=C^\star\) can lie strictly below the maximum quantum energy \(\lambda_{\max}(H_d)\).
Compression also reveals a structural distinction. 
The optimal classical energy threshold $\Cstar$ can lie strictly below the
relaxed quantum optimum $\lambda_{\max}(H_d)$. \revnext{Proposition}~\ref{cor:near} constructs a family of hard instances on which the preparation
barrier persists even though the relaxed quantum optimum exceeds the
classical MaxCut optimum only by a relative $\Theta(1/n)$ factor.
Yet on the same family, the maximally mixed quantum state has energy within a
relative $\Theta(1/n)$ factor of the quantum optimum, but has zero encoded
energy gain and zero decoded mean gain.
Thus neither a nearly tight relaxation nor a \revnext{near-unity} energy
approximation ratio certifies %access to the gain required by the theorem.
the %classically normalized 
encoded energy gain required by the theorem\revnext{.}
We note that this target differs from multiplicative approximations to the
product state or unrestricted optima of the closely \revnext{related} Quantum
MaxCut problem~\cite{Piddock2025,Heilman2026}.

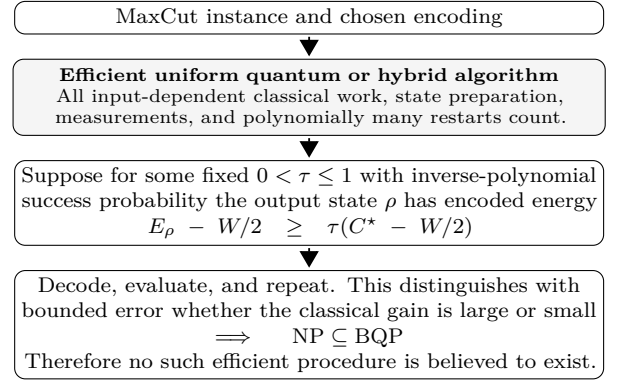
\begin{figure}[t]
\centering
\begin{tikzpicture}[
 stage/.style={draw,rounded corners,align=center,text width=0.88\columnwidth,inner sep=2.6pt,font=\footnotesize},
 inside/.style={draw,rounded corners,align=center,text width=0.88\columnwidth,inner sep=3.2pt,font=\scriptsize,fill=black!4},
 flow/.style={-{Triangle[length=2.2mm,width=2.5mm]},line width=0.7pt,shorten <=0.3mm,shorten >=0.3mm},
 node distance=3.2mm]
 \node[stage] (gap) {MaxCut instance and chosen encoding};
 \node[inside,below=of gap] (prep) {\textbf{Efficient uniform quantum or hybrid algorithm}\\All input-dependent classical work, state preparation, measurements, and polynomially many restarts count.};
 \node[stage,below=of prep] (state) {Suppose for some fixed $0<\tau\le1$ with inverse-polynomial\\ success probability the output state $\rho$ has encoded energy %gain
 \\$E_\rho-W/2\ge\tau(C^\star-W/2)$};
 \node[stage,below=of state] (decision) {Decode, evaluate, and repeat. This distinguishes with\\ bounded error whether \revnext{the} classical gain \revnext{is} large or small\\
%Bounded-error decision of the NP-hard MaxCut-Gain gap\\
$\Longrightarrow\ {\NP\subseteq\BQP}$\\
 Therefore no such %uniformly 
 efficient %algorithm 
 procedure \revnext{is} believed to exist.
%---Decode each output, evaluate its cut, and repeat\\This would distinguish with bounded error whether the\\optimal classical gain is large or small, an NP-hard task\\Therefore $\NP\subseteq\BQP$
};
 \draw[flow] (gap) -- (prep);
 \draw[flow] (prep) -- (state);
 \draw[flow] (state) -- (decision);
\end{tikzpicture}
\caption{%How an all-instance state-preparation guarantee yields the
Complexity consequences of worst-case state preparation performance guarantees for MaxCut-Gain. Here $W$ is the total edge weight,
$C^\star-W/2$ is the optimal classical gain, and
$E_\rho-W/2$ is the encoded energy gain of the prepared quantum state. All
input-dependent preprocessing, %control, 
training, preparation, measurement,
feedback, and repetition count toward the cost. Local decoding and
polynomial repetition would decide the NP-hard MaxCut-Gain gap in
bounded-error quantum polynomial time, implying
$\NP\subseteq\BQP$\revnext{,} which is believed false.
%\revcut{Algorithms involving supplied solution advice, nonuniform or exponentially long controls,}
%\revcut{postselection on exponentially rare events, and claims limited to restricted or typical instances}
%\revcut{fall outside the premise.}
} 
%[prev]From a state-preparation guarantee on every input to the complexity consequence. \revlatest{Here $W$ is the total edge weight and $C^\star$ is the optimal cut value.} Thus $C^\star-W/2$ is the input graph's optimal classical gain, whereas $E_\rho-W/2$ is the prepared state's encoded energy gain. Every resource computed from the input---including seeds, parameters, gauge choices, feedback, measurements, and restarts---is part of the algorithmic cost. If the complete method reaches this encoded target on every valid input with inverse-polynomial success, local decoding and polynomial repetition would give a bounded-error quantum algorithm for the NP-hard task of distinguishing graphs with large versus small optimal classical gain. This would imply $\NP\subseteq\BQP$. Supplied solution advice, input-specific or exponentially long control descriptions, exponentially rare favorable events, and results only on restricted or empirical instance sets do not meet that premise.}
\label{fig:pipeline}
\end{figure}

Graph structure also matters. If the maximum degree of the problem graph is at most $D$, a deterministic classical construction recovers at least $1/(2D-1)$ of the optimal classical gain, so the barrier does not extend unchanged to fixed-degree promise classes~\cite{GutinYeo2023}. Regularity alone does not make MaxCut easy (it remains APX-complete on cubic graphs~\cite{AlimontiKann2000}), but %the present 
our theorem does not establish a preparation threshold for that restriction. Bipartite and bounded-treewidth instances admit efficient exact classical solutions, while our estimate shows that the near-unity mixed state ratio in \revnext{Proposition}~\ref{cor:near} requires $D=\Omega(n)$.

The decoded gain %\revcut{approximation metric} 
\revnext{ratio} metric captures improvement above the random baseline that ordinary %and energy 
approximation ratios can obscure, 
while also providing a generic performance barrier to quantum optimization algorithms. QRAO compression does not remove the barrier, even with vanishing relaxation excess.
%\revlatest{The all-instance barrier already holds in the uncompressed encoding and concerns computational selection, not entanglement. A suitable product state exists for every input.} QRAO compression does not remove the barrier, even with vanishing relaxation excess and a mixed state that is nearly optimal in energy ratio but has zero encoded and decoded mean gain. 
%\revcut{Our results crystalize that quantum optimization algorithm analysis and assessment should therefore consider, report, a compare along multiple dimensions including decoded gain, target reaching probability, readout, precision, and complete end-to-end cost, in addition to standard approximation ratio bounds.}
\revnext{Our results demonstrate that quantum optimization algorithms should be assessed along multiple dimensions. Performance measures such as decoded gain, ordinary approximation ratio, and target-reaching probability should be reported together with training cost, readout and sampling overhead, precision, fault-tolerance overhead, and total end-to-end cost. A clearer understanding of these complexity frontiers can support more informative benchmarks and clarify where quantum advantage is possible.}

The remainder of the paper is structured as follows. 
\revnext{We define the encodings and approximation metrics considered in Sec.~\ref{sec:definitions}. We prove the preparation barrier in Sec.~\ref{sec:preparation} and elucidate %the uncompressed case and 
its algorithmic scope in Sec.~\ref{sec:uncompressed}. We specialize to compressed encodings in Sec.~\ref{sec:compressed}. Finally, we discuss implications and limitations of our results in Sec.~\ref{sec:discussion-shortest}. Technical details appear in Apps.~\ref{app:decoding}, \ref{app:auxiliary}, and \ref{app:padding}.}

% ===== END shortest/introduction.tex =====

\section{Encodings and approximation}\label{sec:definitions}

Table~\ref{tab:measures} distinguishes \revnext{six} performance measures used below. \revnext{They are complementary metrics for algorithm assessment and benchmarking, not interchangeable scores.}

\begin{table*}[t]
\caption{Approximate optimization performance measures used in this work. For a MaxCut instance, $W$ is the total edge weight and $\gstar=\Cstar-W/2$ is the optimal classical gain above a random cut for the input graph. Output expectation values and probabilities include preparation, measurement, decoding, and classical randomness. ``Every input'' refers to a uniform procedure with polynomial end-to-end cost, see Defn.~\ref{def:operational}.  State preparation claims also use the inverse-polynomial operational success condition in Theorem~\ref{thm:prep}.}
\label{tab:measures}
\footnotesize
\begin{tabular*}{\textwidth}{@{\extracolsep{\fill}}p{0.145\textwidth}p{0.195\textwidth}p{0.185\textwidth}p{0.25\textwidth}p{0.125\textwidth}@{}}
\toprule
\raggedright Measure & \raggedright What it reports & \raggedright What it can obscure & \raggedright Complexity statement & \raggedright Main location \tabularnewline
\midrule
\raggedright \revnext{Ordinary approximation ratio} $R=\overline C/\Cstar$ & \raggedright Mean decoded value relative to the classical optimum & \raggedright Random baseline can make $R$ close to one even when no %decoded gain is 
gain recovered & \raggedright MaxCut has no PTAS unless $\mathrm P=\NP$. A quantum PTAS would imply $\NP\subseteq\BQP$ & \raggedright Eqs.~\eqref{eq:gainratio}--\eqref{eq:reversegain}, Prop.~\ref{prop:ordinary-ratio}. \tabularnewline
\addlinespace
\raggedright \revnext{Energy approximation ratio $\Renergy{\rho}=E_\rho/\lambda_{\max}(H_d)$} & \raggedright Mean encoded energy relative to the unrestricted quantum optimum & \raggedright Depends on Hamiltonian offset and need not certify decoded gain & \raggedright Prop.~\ref{cor:near} gives a hard family whose zero-gain maximally mixed state has ratio $1-\Theta(1/n)$. Uniform fixed-fraction gain recovery on this family would imply $\NP\subseteq\BQP$.
& \raggedright Eq.~\eqref{eq:vanish} \tabularnewline
\addlinespace
\raggedright Encoded gain ratio \revnext{$\Renc{\rho}=(E_\rho-W/2)/\gstar$} & \raggedright Quantum state energy above $W/2$, normalized by the %input graph's 
optimal classical gain & \raggedright For compressed encodings it is a relaxation value, can exceed one, and is not itself a feasible cut & \raggedright Theorem~\ref{thm:prep} excludes a fixed positive worst-case (all-input) lower bound unless $\NP\subseteq\BQP$ & \raggedright Eqs.~\eqref{eq:encoded-gains}, \eqref{eq:target} \tabularnewline
\addlinespace
\raggedright Decoded gain ratio \revnext{$\Rdec{\rho}=(\overline C-W/2)/\gstar$} & \raggedright Fraction of the optimal improvement above a random cut recovered by the decoded mean & \raggedright A mean \revnext{value} does not show how probability is distributed among individual outputs & \raggedright Theorem~\ref{thm:prep} and Eq.~\eqref{eq:round} exclude %all-input 
state preparation yielding worst-case fixed positive decoded gain unless $\NP\subseteq\BQP$ & \raggedright Eqs.~\eqref{eq:decoded-gain}, \eqref{eq:gain-chain} \tabularnewline
\addlinespace
\raggedright \revnext{Relaxed-optimum gain ratio $\Rq{\rho}=(E_\rho-W/2)/\gqstar$} & \raggedright \revnext{Fraction of the centered relaxed quantum optimum attained by the state} & \raggedright \revnext{Need not certify a feasible decoded cut or its gain} & \raggedright \revnext{Any fixed positive all-input lower bound implies the target of Theorem~\ref{thm:prep}, since $\gqstar\ge\gstar$} & \raggedright \revnext{Eqs.~\eqref{eq:quantum-gain-fraction}--\eqref{eq:gain-chain}} \tabularnewline
\addlinespace
\raggedright Tail probability $p_\alpha=\Pr[C_G(\mathbf Z)\ge W/2+\alpha\gstar]$ & \raggedright Single-run probability of recovering at least an $\alpha$ fraction of the optimal classical gain & \raggedright Doesn't quantify overshoot, %exponentially 
small tails are not operationally accessible & \raggedright Corollary~\ref{cor:tail} excludes fixed $\alpha>0$ along with inverse-polynomial $p_\alpha$ %on every input 
unless $\NP\subseteq\BQP$ & \raggedright Eq.~\eqref{eq:cantelli}, Cor.~\ref{cor:tail} \tabularnewline
\bottomrule
\end{tabular*}
\end{table*}

% ===== BEGIN stateonly/reordered/model.tex =====
\subsection{MaxCut and its encodings}\label{sec:model}
%For a %finite 
Given a simple graph $G=(V,E,w)$ with positive rational edge weights, let $W=\sum_{ij\in E}w_{ij}>0$. Identify vertex partitions with bit (spin) string assignments $z\in\{-1,1\}^{|V|}$, with cut value %and the optimal cut value by
\begin{equation}
 C_G(z)=\sum_{ij\in E}\frac{w_{ij}}2(1-z_iz_j),
 %\qquad
% \Cstar=\max_{z\in\{-1,1\}^{|V|}} C_G(z).
 \label{eq:classical}
\end{equation}
and maximum cut value $\Cstar:=\max_{z\in\{-1,1\}^{|V|}} C_G(z)$.
A uniformly random classical assignment cuts each edge with probability~$1/2$ and so has expected cut value $E_0=W/2$. Define the gain of an assignment $z$ relative to the random baseline by
\begin{equation}
 g(z)=C_G(z)-\frac{W}{2},
 \label{eq:gain}
\end{equation}
with optimal possible gain %above the random baseline
\begin{equation}
 \gstar=\max_{z\in\{-1,1\}^{|V|}}g(z)=\Cstar-\frac{W}{2}.
 \label{eq:optimal-gain}
\end{equation}
\revlatest{We refer to the subtraction of $E_0$ as \emph{centering}. It shifts the random assignment baseline to zero without changing which assignments maximize the objective or their relative rankings.} Thus the gain $g(z)$ is the improvement of assignment $z$ over the random-assignment mean, and $\gstar$ is the maximum improvement available on the instance. 
As $W>0$, some cut has value strictly greater than the random mean, so $\gstar>0$ always.
\revlatest{Analyzing approximation in terms of this centered quantity assesses what fraction of that available improvement an algorithm recovers. The ordinary approximation ratio instead assesses %the returned value 
solution quality as a fraction of the uncentered optimum $\Cstar$, and is not invariant under linear shifts of the objective function like gain.} 
\revnext{More generally, suitable centering provides a common language for maximization and minimization problems and for positive or signed objectives, settings that require separate conventions or complications for ordinary approximation ratios.}

For encoding on quantum devices, the standard uncompressed Ising representation~\cite{HadfieldHamiltonians2021,lucas2014ising} uses one qubit per vertex and assigns the Pauli observable $P_i=Z_i$ to vertex $i$. \revlatest{We call this the $d=1$ case.} %To extend the same theorem to 
For QRAO compression, we also allow $d\in\{2,3\}$. A valid packing assigns each vertex $i$ to a single-qubit Pauli observable $P_i$, using distinct Pauli axes $X,Y,Z$ for at most $d$ variables on each qubit, and maps adjacent vertices to different qubits. In either encoding, for any density operator $\rho$ define
\begin{equation}
 H_d=\sum_{(ij)\in E}\frac{w_{ij}}2(\Id-dP_iP_j),\qquad
 E_\rho=\Tr(\rho H_d).
 \label{eq:ham}
\end{equation}
Here $\Id$ is the identity operator on the compressed register and $\Tr$ denotes the trace, so $E_\rho$ is the mean energy of quantum state~$\rho$.

Let $\lambda_{\max}(H_d)$ denote the largest eigenvalue of $H_d$. As each Pauli operator product is traceless, the maximally mixed state $I/2^n$, \revnext{where $n$ is the number of qubits,} has mean energy $W/2$, matching the uniformly random baseline.
\revlatest{Using the same baseline as the classical objective, define the \emph{encoded energy gain} of a state and the \emph{relaxed quantum optimum gain} by}
\begin{equation}
 \genc{\rho}=E_\rho-\frac W2,
 \qquad
 \gqstar=\lambda_{\max}(H_d)-\frac W2.
 \label{eq:encoded-gains}
\end{equation}
\revlatest{These are centered Hamiltonian quantities. For $d=1$, the basis state $|z\rangle$ has energy $C_G(z)$, and $\genc{\rho}$ is exactly the mean classical gain obtained by computational basis measurements, with $\gqstar=\gstar$. For $d\in\{2,3\}$, an encoded energy gain is not itself the gain of a feasible cut and may exceed $\gstar$, with }
%For $d=1$, the basis state $|z\rangle$ has energy $C_G(z)$. 
%For $d\in\{2,3\}$, 
tensor products of single-qubit quantum random access code (QRAC) states having %the same energy
energy exactly $C_G(z)$~\cite{Fuller2024}. Hence we have $\lambda_{\max}(H_d)\ge\Cstar$, or equivalently $\gqstar\ge\gstar$, in all three cases.
\revnext{The encoded gain ratio reported in Table~\ref{tab:measures} is $\Renc{\rho}:=\genc{\rho}/\gstar$.}

For a fixed input state $\rho$, let $D_d(\rho)$ denote the random assignment returned by the known local magic-state rounding decoder for $d\in\{2,3\}$~\cite{Fuller2024}, or by computational-basis sampling for $d=1$. \revlatest{Define its \emph{decoded mean gain} and relate it to the encoded energy gain by}
\begin{equation}
 \gdec{\rho}:=\E[C_G(D_d(\rho))]-\frac W2
 =\frac{\genc{\rho}}{d^2}.
 \label{eq:round}
\end{equation}
Here the expectation is taken over the quantum measurement outcomes and any classical randomness used by the decoder. For $d=1$ this %is simply 
reproduces known equality between quantum state energy and the mean sampled cut. \revnext{App.}~\ref{app:decoding} gives the local measurement details for all three encoding settings.
\revlatest{Thus a classically normalized 
encoded performance target $\genc{\rho}\ge\tau\gstar$ yields decoded mean gain $\gdec{\rho}\ge\tau\gstar/d^2$. The normalizer $\gstar$ %belongs to 
comes from the original MaxCut input in both expressions. Below, any unqualified use of ``gain'' refers to a classical assignment or decoded output; Hamiltonian quantities are referred to as encoded energy gain or relaxed quantum optimum gain.}

% ===== END stateonly/reordered/model.tex =====

% ===== BEGIN stateonly/reordered/scales.tex =====
\subsectionview{scales}{Gain versus \revnext{ordinary approximation ratio}}
Consider a uniform quantum or hybrid optimization algorithm, as we formalize shortly in Definition~\ref{def:operational} below. 
Let $\mathbf Z$ denote the random assignment returned by the complete algorithm and readout procedure, and let $z$ denote one sample realization. Define the mean returned value and the recovered fraction of the available gain by
\begin{equation}
 \overline C=\E[C_G(\mathbf Z)],
 \qquad
 \revnext{\Rdec{\rho}}=\frac{\overline C-W/2}{\gstar}
 =\frac{\E[g(\mathbf Z)]}{\gstar}.
 \label{eq:decoded-gain}
\end{equation}
Unless stated otherwise, expectations and probabilities involving $\mathbf Z$ are over all randomness in the end-to-end procedure, including classical random choices, randomized state preparation, quantum measurements, and decoding. We call \revnext{$\Rdec{\rho}$ the \emph{decoded gain ratio} (fraction)}.  

Also set $c=\Cstar/W>1/2$ and $R=\overline C/\Cstar$, where $R$ is the ordinary expected approximation ratio of the decoded output. When state energy is compared with the relaxed quantum optimum, we call \revnext{$\Renergy{\rho}:=E_\rho/\lambda_{\max}(H_d)$} the energy approximation ratio. \revlatest{The centered analogue is the \revnext{relaxed-optimum gain ratio}}
\begin{equation}
 \revnext{\Rq{\rho}}=\frac{\genc{\rho}}{\gqstar}.
 \label{eq:quantum-gain-fraction}
\end{equation}
\revlatest{Unlike the theorem's target $\tau$, which is normalized by the classical optimum gain $\gstar$, \revnext{$\Rq{\rho}$} is normalized by the relaxed quantum optimum gain $\gqstar$. For a fixed state followed by decoding $D_d$, the three normalized gain measures are related by}
\begin{equation}
 \revnext{\Rdec{\rho}}=\frac{\gdec{\rho}}{\gstar}
 =\frac{1}{d^2}\frac{\genc{\rho}}{\gstar}
 =\frac{\revnext{\Rq{\rho}}}{d^2}\frac{\gqstar}{\gstar}.
 \label{eq:gain-chain}
\end{equation}
\revlatest{Thus the encoded target $\genc{\rho}/\gstar\ge\tau$ of the theorem guarantees \revnext{decoded gain ratio $\Rdec{\rho}\ge\tau/d^2$}. It is not the same as requiring \revnext{$\Rq{\rho}\ge\tau$}, although the latter is at least as strong because $\gqstar\ge\gstar$.} \revcut{The exact conversion between decoded gain ratio and ordinary approximation ratio is shown below. If $\Rdec{\rho}\ge\alpha$ for $0\le\alpha\le1$, then}\revnext{The ordinary approximation ratio includes the random-cut baseline, so a decoded gain bound also bounds $R$. If $\Rdec{\rho}\ge\alpha$ for $0\le\alpha\le1$, then}
\begin{equation}
 R=\revnext{\Rdec{\rho}}+\frac{1-\revnext{\Rdec{\rho}}}{2c} \;\ge\frac{1+\alpha}{2}.
 \label{eq:gainratio}
\end{equation}
\revcut{with the inequality following from $0\le\alpha\le1$ and $\Rdec{\rho}\ge\alpha$.}\revcut{The inequality uses $c\le1$, $\Rdec{\rho}\le1$, and $\Rdec{\rho}\ge\alpha$.}\revnext{Because $c\le1$ and $\Rdec{\rho}\le1$, the exact expression for $R$ is at least $(1+\Rdec{\rho})/2$. The assumed gain bound then gives the inequality.}
%gives 
%\begin{equation}
% R\ge\frac{1+\alpha}{2}.
% \label{eq:gainratio2}
%\end{equation}

Conversely, an $a$-approximation ($0<a\le1$), combined with a classical cut of value at least $W/2$, certifies only
\begin{equation}
 \revnext{\Rdec{\rho}}\ge\max\!\left\{0,\frac{ac-1/2}{c-1/2}\right\}.
 \label{eq:reversegain}
\end{equation}
The classical Goemans--Williamson (GW) algorithm for MaxCut applies random-hyperplane rounding to a semidefinite relaxation and guarantees, in expectation over the random hyperplane, a cut of value at least $a_{\rm GW}\Cstar$, where
$a_{\rm GW}=\min_{0<\theta\le\pi}2\theta/[\pi(1-\cos\theta)]\approx0.878567$~\cite{Goemans1995}, where $\theta$ is the angle between the two %semidefinite-program 
real vectors associated with an edge. 
%For comparison, for every fixed $\delta>0$, achieving an ordinary approximation ratio of $16/17+\delta\approx0.941+\delta$ is NP-hard~\cite{Hastad2001}. Assuming the Unique Games Conjecture, the NP-hardness threshold tightens to $a_{\rm GW}+\delta$~\cite{KhotKindlerMosselODonnell2007}. 
%
Observe that Eq.~\eqref{eq:reversegain} gives no positive \revnext{decoded gain ratio} \emph{\revnext{from the ordinary approximation-ratio bound alone}} when $c\le1/(2a_{\rm GW})\approx0.569108$. Actual GW cuts may be better. On promises $c\ge0.6$, for example, the bound certifies \revnext{$\Rdec{\rho}\ge0.2714$}. A stronger classical gain analysis recovers $\Omega(1/\log(1/\varepsilon))$ of the available gain when $c=1/2+\varepsilon$~\cite{CharikarWirth2004,KhotODonnell2009}. Clearly this fraction can vanish as $\varepsilon$ approaches zero. %Our result does not determine the optimal ordinary approximation ratio.

\paragraph*{\revnext{Hardness from ordinary approximation ratio.}}
For MaxCut %the %worst-case approximation c 
its hardness of approximation in terms of \revnext{ordinary worst-case approximation ratio} is well characterized. 

\begin{proposition}[Standard ordinary-ratio barrier]
\label{prop:ordinary-ratio}
Let $a_{R}=16/17$. For any fixed
$0<\delta\le1-a_{R}$, suppose that a uniformly efficient
quantum or hybrid procedure outputs, on every unweighted MaxCut
instance, a cut $\mathbf Z$ satisfying
\[
 \mathbb E[C_G(\mathbf Z)]
 \ge (a_R+\delta)C^\star .
\]
Then $\NP\subseteq\BQP$. Assuming the Unique Games Conjecture,
the same conclusion holds with \revcut{$a_r$}\revnext{$a_R$} replaced by
$a_{\rm GW}\approx0.878567$
\cite{Hastad2001,KhotKindlerMosselODonnell2007}.
\end{proposition}

\begin{proof}
Let $r_0$ denote the relevant hardness threshold. Since
$C_G(\mathbf Z)\le C^\star$,
\[
 \Pr\!\left[
 C_G(\mathbf Z)\ge(r_0+\delta/2)C^\star
 \right]
 \ge
 \frac{\delta/2}{1-r_0-\delta/2}.
\]
This probability is constant. Polynomial repetition and retention
of the best observed cut therefore give a bounded-error quantum
$(r_0+\delta/2)$-approximation, contradicting the corresponding
NP-hardness result unless $\NP\subseteq\BQP$.
\end{proof}
%The Proposition applies to both standard and compressed encodings.
%
Thus MaxCut has no (ordinary approximation ratio) polynomial-time approximation scheme (PTAS) unless $\mathrm{P}=\NP$. A bounded-error randomized or quantum PTAS would imply $\NP\subseteq\mathrm{BPP}$ or $\NP\subseteq\BQP$, respectively, where $\mathrm{BPP}$ denotes the class of bounded-error probabilistic polynomial time problems. Under the Unique Games Conjecture, the GW algorithm is an optimal constant-factor approximation algorithm for MaxCut in terms of worst-case performance.

Summarizing, the different notions of gain complement, rather than replace, the ordinary approximation ratio in assessing quantum optimization algorithms. %The baseline and readout must remain explicit.
Using gain we now state our general state preparation barrier.

\section{State preparation barrier from gain}\label{sec:preparation}

% ===== BEGIN shared/preparation.tex =====
Corollary~1.13 of Ref.~\cite{Dinur2025} states the following complexity gap assuming its Hypothesis~3.6, %The final article records that Khot, Minzer, and Safra subsequently proved that hypothesis~\cite{KhotMinzerSafra2023}.} 
which was subsequently proven in Ref.~\cite{KhotMinzerSafra2023}. 
Combining the two gives an unconditional statement: there exists %an absolute 
a universal constant $c_1>0$ such that, for every fixed sufficiently small %fixed 
$\varepsilon>0$, for unweighted graphs it is NP-hard to distinguish between cases 
\begin{equation}
 \begin{aligned}
 \text{YES:}\quad &\Cstar/W\ge\tfrac12+\varepsilon,\\
 \text{NO:}\quad &\Cstar/W\le\tfrac12+\eta_\varepsilon,
 \end{aligned}
 \label{eq:gap}
\end{equation}
with $\eta_\varepsilon=c_1\varepsilon/\log(1/\varepsilon)$ and here $W=|E|$. The cited corollary measures each cut as a fraction of the total number of edges $C_G(z)/W$, and maximizing over $z$ gives Eq.~\eqref{eq:gap}; parallel edges, if any, can be merged into positive integer weights. Its polynomial-time reduction therefore supplies instances in our input class. The YES and NO cases have optimal gain at least $\varepsilon W$ and at most $\eta_\varepsilon W$, respectively. 
%The ratio $\eta_\varepsilon/\varepsilon=c_1/\log(1/\varepsilon)$ tends to zero. 
This separation gives the uncompressed result directly as a quantum algorithm producing a fixed fraction of optimal gain can be used to  decide this problem.  For the compressed encodings, choosing $\varepsilon$ sufficiently small leaves the YES mean above every possible NO output even after the decoder loses a fixed factor.

\begin{definition}[Uniform operational preparation]
\label{def:operational}
The input $I$ specifies the graph and a valid packing, and $N=|I|$ is its bit length. A uniform operational preparation procedure has a polynomial-time classical controller that, from~$I$ and without solution advice, specifies a polynomial-size adaptive quantum circuit over a fixed universal gate set. It may interleave polynomial-time classical computation, randomness, quantum gates, and measurements, and it returns a quantum register together with a finite classical record $\mathcal R$. All preprocessing, control and parameter selection, state generation, and repetitions count toward the cost. Here adaptive means that later gates, measurements, and
polynomial-time classical control may depend on the input and on
earlier classical randomness and measurement outcomes. %It does not permit solution advice or nonuniform or exponentially long control descriptions.

The procedure succeeds with probability at least $s(N)$ if some event $\mathcal E$ determined by $\mathcal R$ has $\Pr(\mathcal E)\ge s(N)$ and the conditional output state $\rho_{\mathcal E}$ satisfies the stated target. 
This probability is over all classical randomness and quantum measurement outcomes generated by the procedure. %An arbitrary ensemble decomposition of a density operator does not define such an event. 
The event $\mathcal E$ must correspond to a set of actual run
histories recorded in $\mathcal R$, and cannot be introduced merely
by choosing a favorable ensemble decomposition of the output density
operator. %The procedure need not identify the histories in
%$\mathcal E$ as successful. The reduction decodes every output and
%checks the resulting cut value, so these histories contribute their
%probability $\Pr(\mathcal E)$ without postselection.
%
%A success flag is unnecessary because the reduction in the proof of Theorem~\ref{thm:prep} never conditions on $\mathcal E$. 
\end{definition}
The argument of Proposition~\ref{prop:ordinary-ratio} also applies
when its mean guarantee holds conditionally on an operational event
$\mathcal E$ of inverse-polynomial probability, with the unconditional
probability %of obtaining a qualifying cut %simply acquires the
reduced by a factor 
$\Pr(\mathcal E)$. %and remains inverse polynomial.

Broad examples to which Definition~\ref{def:operational} applies include polynomial resource cost implementations of (i) QAOA with efficient parameter selection, (ii) finite-time quantum annealing or adiabatic evolution, (iii) %adaptive open-system or measurement-based routines, 
adaptive open-system or measurement-based routines with efficiently specified feedback, and (iv) decoded quantum interferometry, among other schemes.  Generally deterministic or randomized preprocessing, classical feedforward, and polynomially many restarts are allowed. Solution advice, exponentially specified controls, and conditioning on exponentially rare events are not. Sec.~\ref{sec:uncompressed} gives more precise qualifications for these examples.

\begin{theorem}[Algorithm-independent preparation barrier]
\label{thm:prep}
Fix a constant $0<\tau\le1$. For the standard uncompressed encoding $d=1$ case, a uniform operational preparation procedure that, on every valid instance, prepares a state whose encoded energy gain satisfies
\begin{equation}
 \genc{\rho}=E_\rho-W/2\ge\tau\gstar
 \label{eq:target}
\end{equation}
with operational success probability $s(N)\ge1/\poly(N)$ would imply $\NP\subseteq\BQP$. 
For $d=1$, computational-basis readout gives $\gdec{\rho}=\genc{\rho}$, so the premise is also a decoded mean-gain guarantee.
%Here $\tau$ is measured relative to the input graph's optimal classical gain $\gstar$, not the relaxed quantum-optimum gain $\gqstar$. 
%Any guarantee with $\tau>1$ also meets the target at $\tau=1$. 

The same conclusion holds for each compressed QRAO encoding $d\in\{2,3\}$, and remains true under  restricted families of fully occupied packings where every qubit contains exactly $d$ variables. 
For the classical-threshold target $\tau=1$, this includes any guarantee $E_\rho\ge T$ with $\Cstar\le T\le\lambda_{\max}(H_d)$. 
%There any guarantee with $\tau>1$ also implies the conclusion at $\tau=1$. 
For each fixed pair $(d,\tau)$, the family of hard instances may depend on $(d,\tau)$. %The conclusion also holds when every qubit contains exactly $d$ variables.
\end{theorem}
No ansatz, optimizer, locality, or control-count restriction is imposed for the uniform polynomial-time quantum computation of the theorem. Figure~\ref{fig:pipeline} above summarizes the reduction %of the theorem 
and its operational boundary. %and the reduction that makes it consequential. 
\begin{proof}
Choose fixed $\varepsilon$ with $\eta_\varepsilon<\tau\varepsilon/(2d^2)$ and set $b=\tau\varepsilon/d^2$. Conditioned on the successful preparation event~$\mathcal E$ for a YES instance, let $\mathbf Z$ be the decoded assignment and set $Y=C_G(\mathbf Z)/W\in[0,1]$. The expectation and probability in the following bound are conditional on $\mathcal E$ and are over the decoder measurements and classical randomness. Then $\E[Y]\ge1/2+b$. 
Since $0\le Y\le1$,
\[
 \mathbb E[Y]
 \le t\Pr(Y\le t)+\Pr(Y>t)
 =t+(1-t)\Pr(Y>t).
\]
Hence for $t=1/2+b/2$ we have
\begin{equation}
 \Pr(Y>t)
 \ge\frac{\mathbb E[Y]-t}{1-t}
 \ge b/2.
 \label{eq:tail}
\end{equation}
Thus each run has cut value that exceeds~$tW$ with probability at least $s(N)b/2$. No cut of a NO instance exceeds~$tW$, including outcomes from unsuccessful preparations. Repeating $O(1/(s(N)b))$ times and accepting if any cut exceeds the threshold gives bounded error. A rational threshold strictly inside the gap permits exact comparison because all input weights are rational. Neither a success flag nor knowledge of $\Cstar$ is required.

For exact occupancy, take $d$ disjoint copies of the source graph and pack corresponding vertices onto the $d$ axes of one qubit. Both $W$ and $\Cstar$ are multiplied by $d$, %so the gain gap is unchanged. 
so the ratio $C^\star/W$ and the normalized YES/NO promise gap in
Eq.~\eqref{eq:gap} are unchanged.
The reduction and the local decoder have polynomial cost. 
%\revlatest{The classical gap is NP-hard under polynomial-time many-one promise reductions. Composing such a reduction with this bounded-error quantum decider puts every language in $\NP$ inside $\BQP$.}

Finally, the promise problem in Eq.~\eqref{eq:gap} is NP-hard under
polynomial-time many-one reductions, i.e., every problem in $\NP$ can be
mapped efficiently to a graph satisfying either its YES or its NO
condition. The bounded-error quantum procedures constructed above
would distinguish these two cases. Hence composing the mapping with that
procedure would solve every problem in $\NP$ within
$\BQP$, implying $\NP\subseteq\BQP$ as claimed.
\end{proof}
The theorem also applies when the procedure chooses its own valid packing. It does not require the preparation method %itself 
to use
a particular QRAO decoder. Once the state has been prepared, the
bounded-error quantum algorithm constructed in the proof can apply
the fixed efficient decoder. %instead of any readout proposed by the
%method. The reduction may append the decoder regardless of the procedure's intended readout. 
%Moreoever, no ansatz, optimizer, locality, or control-count restriction is imposed within uniform polynomial-time quantum computation.
App.~\ref{app:decoding} considers implementation error and expected runtime.

The gap also bounds worst-case gain as $\varepsilon$ varies. Fix any sufficiently small $\varepsilon>0$ and any constant $b>c_1$. If a uniform polynomial-cost quantum pipeline, on every graph with $\Cstar/W\ge1/2+\varepsilon$, returned a cut with conditional mean at least $W[1/2+b\varepsilon/\log(1/\varepsilon)]$ on an operational event of inverse-polynomial probability, its output would distinguish the YES and NO cases of Eq.~\eqref{eq:gap} after polynomially many repetitions, implying $\NP\subseteq\BQP$. Classical algorithms achieve gain $\Omega(\varepsilon W/\log(1/\varepsilon))$ on this promise~\cite{CharikarWirth2004}. Hence, up to constant factors, $\varepsilon W/\log(1/\varepsilon)$ is the optimal worst-case gain scale on this promise for uniform polynomial-cost methods, assuming $\NP\not\subseteq\BQP$. This promise-specific upper bound does not follow from the fixed-$\tau$, all-instance barrier of Theorem~\ref{thm:prep} alone.

\subsection{Tail probability barriers}
%The energy $E_\rho=\Tr(\rho H_d)$ is a Born-rule expectation, while the decoded mean $\E[C_G(\mathbf Z)]$ is a first moment of the complete output distribution. Equal 
Mean values, while helpful, can mask concentrated mediocre outcomes or useful probability distribution tails. Eq.~\eqref{eq:tail} turns a high expectation value into detectable tail mass, but failing its premise does not exclude rare or good samples. 
On studied Ising and random-MaxCut instances, QAOA output can resemble pseudo- or approximate-Boltzmann distributions, linking an effective temperature parameter to high-quality sample probabilities~\cite{DiezValle2023,Lotshaw2023}. Such fits do not imply useful tail mass without considering normalization and density of states. Gibbs and conditional value at risk (CVaR) loss functions instead propose nonlinear or tail-restricted functionals (i.e., the upper tail for MaxCut)~\cite{Li2020Gibbs,Barkoutsos2020}. Together these approaches further motivate distribution-sensitive 
algorithm performance reporting. %assessment. 
%reporting. %Theorem~\ref{thm:prep} applies when their outputs meet its mean premise. Otherwise only the decoded tail probability is constrained.} We do not claim that mean optimization is preferable.

On the other hand, concentration effects can in some cases significantly diminish %the effects of 
higher distributional moments beyond the mean. 
%\revlatest{The one-sided Chebyshev--Cantelli inequality gives a direct concentration bound.} 
For $0<\alpha\le1$, let $T_\alpha=W/2+\alpha\gstar$ and $\mu=\E[C_G(\mathbf Z)]$, and suppose $\Delta=T_\alpha-\mu>0$. %With $\operatorname{Var}$ denoting variance, and with variance and probability taken over the same complete output distribution as this expectation, 
The Chebyshev--Cantelli inequality~\cite{Savage1961} then gives
\begin{equation}
 \Pr[C_G(\mathbf Z)\ge T_\alpha]
 \le \frac{\operatorname{Var}[C_G(\mathbf Z)]}
 {\operatorname{Var}[C_G(\mathbf Z)]+\Delta^2}.
 \label{eq:cantelli}
\end{equation}
Hence a gap $\Delta=\Omega(\gstar)$ and standard deviation $o(\gstar)$ imply $o(1)$ probability mass above the target. Known concentration results establish such control only in structured regimes, including certain shallow-circuit observables and dense-QAOA evolutions under explicit depth and norm conditions~\cite{AnshuMetger2023}. They do not apply to arbitrary states in this theorem. Concentration of a fixed-parameter expectation across random instances also does not control shots required on a single input.

Define $p_\alpha=\Pr[C_G(\mathbf Z)\ge W/2+\alpha\gstar]$, with the probability again taken over the complete output distribution.
\begin{corollary}[Tail-probability barrier]
\label{cor:tail}
Fix a constant $0<\alpha\le1$. If a uniform quantum or hybrid procedure with polynomial end-to-end cost returns, on every valid MaxCut input, a cut $\mathbf Z$ satisfying
\begin{equation}
 C_G(\mathbf Z)\ge W/2+\alpha\gstar
 \label{eq:tail-target}
\end{equation}
with probability $p_\alpha\ge1/\poly(N)$, then $\NP\subseteq\BQP$.
\end{corollary}
%\revlatest{App.~\ref{app:auxiliary} gives the proof.}
\begin{proof}
Choose fixed $\varepsilon>0$ with $\eta_\varepsilon<\alpha\varepsilon/2$. On a YES instance of the gap in Eq.~\eqref{eq:gap}, the event in Eq.~\eqref{eq:tail-target} gives a cut of value at least $W/2+\alpha\varepsilon W$. On a NO instance, every cut has value at most $W/2+\eta_\varepsilon W$. Choose a fixed rational number strictly between $\eta_\varepsilon$ and $\alpha\varepsilon$. Polynomial repetition and exact comparison with the resulting rational threshold give a bounded-error quantum decision procedure for the gap.
\end{proof}
%Report $p_\alpha$ together with the sampling and postprocessing cost. 
%Because $p_\alpha$ is a per-run success probability, report it
%together with the cost per run and the number of repetitions. In
%particular, obtaining a target-reaching cut with constant
%probability requires $O(1/p_\alpha)$ independent runs, in addition
%to the classical cost of decoding, repair, and objective evaluation.
%Corollary~\ref{cor:tail} 
The corollary requires no mean value guarantee. It excludes an inverse-polynomial upper tail, not exponentially small mass, and is again worst case. Random-instance claims require a specified ensemble and scaling; Erd\H{o}s--R\'enyi, random-regular, planted, weighted, and application-derived inputs need not behave alike. The% theorem does 
results do not rule out good performance with high probability over such an ensemble. By contrast, typical-output claims concern repeated runs on one specified problem instance. That setting is described by $p_\alpha$ and the distribution of decoded cut values. Such a random-instance barrier would require distributional hardness or a worst-to-average-case reduction for the chosen ensemble. %, neither assumed here.

% ===== END shared/preparation.tex =====

%\section{Uncompressed encoding and algorithmic scope}\label{sec:uncompressed}
\section{Algorithmic scope}\label{sec:uncompressed}

% ===== BEGIN shortest/uncompressed.tex =====

Here we discuss implications of our result for algorithms based on standard uncompressed encodings, and specialize to the compressed case in Sec.~\ref{sec:compressed}.

With $d=1$, the standard encoding uses one qubit per vertex. The objective Hamiltonian $H_1$ is diagonal with $\lambda_{\max}(H_1)=\Cstar$, its eigenstates %are spanned by 
correspond to efficiently preparable product states~$|z\rangle$, and %Computational basis measurement gives $\gdec{\rho}=\genc{\rho}$, so 
the encoded and decoded mean gains coincide \branchadd{under computational-basis readout:} $\gdec{\rho}=\genc{\rho}$. Eq.~\eqref{eq:target} %therefore asks directly for the same 
asks for a fraction $\tau$ of the available classical gain, %with At $\tau=1$, support is 
restricted to optimal cuts at $\tau=1$. %Good states are nevertheless simple product states. Preparing $|z^\star\rangle$ is trivial once an optimal assignment $z^\star$ is known. Theorem~\ref{thm:prep} constrains selecting such a state from the instance, not the entanglement or depth needed to prepare a known answer.

\subsection{Conditional algorithmic coverage}
The theorem applies to an end-to-end algorithm only when that method satisfies Definition~\ref{def:operational}. For QAOA~\cite{Farhi2014,Hadfield2019}, assuming $\NP\not\subseteq\BQP$, it rules out a uniform polynomial-time procedure that finds parameters and executes a corresponding polynomial-size quantum circuit attaining Eq.~\eqref{eq:target} on every possible input graph, irrespective of the particular mixer operator used or many other algorithm variations. The statement constrains efficiently finding and implementing satisfactory parameters, not whether such parameters can exist. This is distinct from NP-hardness of classically training cases of variational quantum algorithms~\cite{Bittel2021}. A related companion preprint proves worst-case \#P-hardness of exact or exponentially precise expectation value evaluation for %a specified MaxCut 
QAOA circuit at any fixed depth $p\ge2$~\cite{HadfieldQAOAExpectation2026}. %\revlatest{Those results concern parameter optimization and fixed-circuit evaluation. The present barrier applies to any end-to-end preparation procedure satisfying Definition~\ref{def:operational}, regardless of ansatz or mixer.} 
%
%Complementary QAOA results establish worst-case limitations when
%shallow circuits access only local graph neighborhoods, when
%symmetry preservation limits the achievable approximation ratio,
%and when a desired guarantee requires many QAOA rounds
%\cite{Farhi2020,Bravyi2020,Benchasattabuse2025}. 
Complementary QAOA results establish worst-case limitations when
shallow implementations access only local graph neighborhoods, or when
symmetry preservation limits the achievable approximation ratio
\cite{Farhi2020,Bravyi2020}. For QAOA with the Grover-mixer, %complementary
recent work derives nontrivial %approximation quality dependent 
lower bounds relating the number of QAOA layers to approximation quality, %including
%$p=\Omega(\sqrt{|E|})$ for any fixed ratio above $1/2$ on
%bipartite MaxCut instances; 
though the same framework gives only trivial
bounds for standard transverse-field QAOA on local MaxCut
Hamiltonians~\cite{Benchasattabuse2025}.

Hybrid schemes that incorporate QAOA or other ans\"atze as subroutines such as iterative quantum optimization~\cite{Bravyi2020,dupont2023quantum,Brady2024,FinzgarQIRO2024,Brady2026} or quantum relax-and-round~\cite{dupont2024extending, maciejewski2024multilevel} 
 instead use quantum measurement information to  guide choices within classical algorithm frameworks, offering attractive alternatives. Nevertheless the theorem also reaches these approaches %only 
 through any end-to-end MaxCut guarantee satisfying Definition~\ref{def:operational}, when applicable. 
%
%\revlatest{Complementary work establishes QAOA-specific depth, locality, and symmetry limitations. Quantum-informed recursive optimization (QIRO) and iterative quantum optimization (IQO) instead place quantum information inside recursive classical reductions~\cite{Farhi2020,Bravyi2020,Benchasattabuse2025,FinzgarQIRO2024,Brady2024,Brady2026}. These results concern particular ansatzes or problems. The present theorem reaches them only through an end-to-end MaxCut guarantee satisfying Definition~\ref{def:operational}.}

Warm start approaches in particular make the advice boundary explicit. Standard variants initialize %QAOA 
a quantum ansatz using a classical relaxation or incumbent solution~\cite{Egger2021}, and iterative extensions update the incumbent after repeated sampling~\cite{LotshawIWS2026,MaciejewskiNDAWS2026}. If a supplied cut $z_0$ satisfies $C_G(z_0)-W/2\ge\tau\gstar$, then $|z_0\rangle$ for $d=1$, or its QRAC product embedding for $d\in\{2,3\}$, already satisfies Eq.~\eqref{eq:target} and gives no contradiction. The theorem therefore does not (cannot) exclude supplied advice, promises that %efficiently provide 
such a seed is provided, or empirical incumbent improvements, including those obtained leveraging system noise~\cite{MaciejewskiNDAR2025,MaciejewskiNDAWS2026}. Seeds and updates computed from $I$ count toward the end-to-end cost, so any uniform polynomial cost all-instance %state 
guarantee remains covered. %\revlatest{An incumbent-only loop need not satisfy the state premise mid-run. If it has the corresponding uniform polynomial-cost final-cut guarantee, the classical gap applies directly.}

%\revlatest{Noise-directed adaptive remapping (NDAR) aligns the current best string with a device-noise attractor through an invertible bit-flip gauge~\cite{MaciejewskiNDAR2025}. Noise-directed adaptive warm-starting (ND-AWS) adds iterative bias and reports hardware improvements over a nongauge control~\cite{MaciejewskiNDAWS2026}.} Here adaptive means that information obtained in one round is used
%to update the mapping or warm-start bias for the next round. Incumbents and gauges generated within the run count toward its end-to-end cost. A polynomial-cost method whose final state, after undoing the gauge, meets Eq.~\eqref{eq:target} on every instance remains covered. Device-specific improvements are empirical and should be benchmarked under matched device, layout, noise, and calibration against %nonadaptive, nongauge, and random-gauge controls.
%fixed-gauge, fixed-bias, nongauge, and random-gauge controls.

%Several QAOA analyses 
In terms of existing quantum optimization analyses, a number of QAOA results in the literature %can be read 
directly translate to gain above random. For every unweighted triangle-free $D$-regular graph, $D\ge2$, the standard depth-one QAOA state $|\psi_1(\gamma,\beta)\rangle=e^{-i\beta B}e^{-i\gamma H_1}|+\rangle^{\otimes |V|}$, with mixer $B=\sum_iX_i$, has expected energy $E_1(\gamma,\beta)=\langle\psi_1(\gamma,\beta)|H_1|\psi_1(\gamma,\beta)\rangle$. Let $E_1^\star=\max_{\gamma,\beta}E_1(\gamma,\beta)$. Optimizing its two angles gives~\cite{Wang2018}
\begin{equation}
 \frac{E_1^\star-W/2}{W}
 =\frac1{2\sqrt D}\left(1-\frac1D\right)^{(D-1)/2},
  \label{eq:qaoa-shortest}
\end{equation}
with $W=|E|$, or in terms of optimal classical gain
\begin{equation}
 \frac{E_1^\star-W/2}{\gstar}
 \ge\frac1{\sqrt D}\left(1-\frac1D\right)^{(D-1)/2}.
 \label{eq:qaoa-shortest2}
\end{equation}
The %second relation 
inequality uses $\gstar\le W/2$, and its right-hand side %is asymptotic to 
approaches $1/\sqrt{eD}$ from above as $D$ grows. 
%\revlatest{One optimum is $\beta=\pi/8$ and $\gamma=\arctan(1/\sqrt{D-1})$. 
Equality %in the gain bound 
holds for bipartite complete graphs $K_{D,D}$, where $\gstar=W/2$. At $D=100$, the mean recovers about $6.08\%$ of the available gain for $K_{100,100}$ and at least this fraction for every other unweighted triangle-free $100$-regular graph. \branchadd{For the same triangle-free graph class, classical polynomial-time algorithms also obtain gain $\Omega(W/\sqrt D)$~\cite{Barak2015}. The depth-one quantum bound has the same order in $D$, without an end-to-end advantage claim.} Depth-two analysis for high-girth graphs similarly yields explicit gain bounds%and compare QAOA with local classical algorithms
~\cite{Marwaha2021}, %Related QAOA and classical bounds treat
and gain results have also been obtained for Max $k$XOR%, including MaxCut at $k=2$
~\cite{MarwahaHadfield2022}. 
%These are structured-family gain guarantees, not hardness witnesses, consistent with the fixed-degree boundary above.
%At fixed $D$, 
These results are positive gain guarantees for structured
graph families, not hardness results, and in particular they are consistent with the %earlier 
observation that the all-instance barrier does not extend unchanged to fixed-degree promise classes.

Decoded quantum interferometry (DQI) provides a structurally different example. It combines quantum Fourier transforms with reversible classical decoding to bias samples toward high objective values, with positive results obtained for structured optimization problems~\cite{JordanDQI2025}. This approach does not evade the present worst-case statement; any uniform polynomial-cost DQI procedure attaining Eq.~\eqref{eq:target} on every MaxCut instance would satisfy the hypothesis of Theorem~\ref{thm:prep}. Existing DQI guarantees instead exploit promise structure. In the MaxCut regime where DQI has a nontrivial asymptotic guarantee, the corresponding instances are exactly solvable in classical polynomial time~\cite{ParekhDQIMaxCut2025}. DQI therefore illustrates the structured-promise route left open by the theorem rather than a competing all-instance claim.

Other quantum optimization approaches are likewise covered by the theorem only when their complete implementation satisfies Definition~\ref{def:operational}.
%The same qualification matters for other models. 
Representative examples include quantum annealing and adiabatic optimization~\cite{Kadowaki1998,FarhiAdiabatic2000}, 
approaches based on imaginary time evolution or dissipation~\cite{Kliesch2011,CubittDQE2023,DingChenLin2024,bauer2024combinatorial,morris2024performant,StollenwerkHadfield2025}. 
%dissipative simulation and state preparation, including the dissipative quantum eigensolver~\cite{Kliesch2011,CubittDQE2023,DingChenLin2024}, and measurement-based imaginary-time or approximate optimization~\cite{StollenwerkHadfield2025}.
 
An annealing, adiabatic, open-system, measurement-based, or hybrid method is covered when its complete finite-time implementation is uniformly realizable by a polynomial-size adaptive circuit to the accuracy required in App.~\ref{app:decoding}. %The theorem does not assert that every such evolution has this property~\cite{Aharonov2007,Cleve2017}.} 
Whether this condition holds is model dependent, with known equivalence and simulation results requiring explicit assumptions on the evolution and its implementation~\cite{Aharonov2007,Cleve2017}. 
Specifics such as duration, control construction, precision, success probability, resets, and repetitions all %contribute to each
count in the assessment.

\subsection{Ordinary approximation ratio can be misleading}
The distinction between ordinary approximation ratio and gain is
already visible on easy graphs. For $K_{100}$, $W=4950$ and
$C^\star=2500$, so the random baseline $W/2=2475$ already has ordinary approximation ratio $0.99$ while at the same time recovering zero gain. 
%corresponds to ordinary approximation ratio $0.99$, despite having
%zero gain. 

To hold both the optimum and the mean cut value fixed
while changing the baseline, consider $K_{50,50}$, for which
$W=C^\star=2500$. Outputting for example the bipartition with probability $0.99$
and the all-equal cut otherwise again gives mean cut value $2475$
and ordinary approximation ratio $0.99$, but now gives gain
%approximation ratio
\[
\frac{2475-1250}{2500-1250}=0.98.
\]
Hence the same mean value and ordinary approximation ratio can
represent either no improvement over random assignment, or recovery
of nearly all available gain.

There are also structured families where obtaining positive gain is easy. If the maximum vertex degree is at most~$D$, a deterministic polynomial-time algorithm~\cite{GutinYeo2023} finds a matching $\mathcal M$ with total weight $w(\mathcal M)\ge W/(2D-1)$ and an expected cut satisfying
\begin{equation}
 \overline C\ge\frac W2+\frac{w(\mathcal M)}2,
 \qquad
 \frac{\overline C-W/2}{\gstar}\ge\frac1{2D-1}.
 \label{eq:bounded-shortest}
\end{equation}
App.~\ref{app:auxiliary} gives the %matching 
construction and a derandomization. 
%The %all-instance 
%general obstruction therefore should not be imported unchanged into every fixed-degree promise class.

% ===== END shortest/uncompressed.tex =====

%\section{Compression and a weaker preparation target}
\section{Compressed encodings and classical gain targets} \label{sec:compressed}
For QRAO with $d\in\{2,3\}$, each classical assignment has a QRAC product state embedding with the same objective value, but the quantum optimum can be larger than the maximum cut value. Consequently, preparing a state with energy at least $C^\star$ can
be strictly easier than maximizing the compressed Hamiltonian.

QRAO %variable placement 
compression corresponds to the $k=1$ case of the
Pauli-correlation encoding (PCE) framework, which stores
$m=O(n^k)$ bits in signs of weight-$k$ Pauli expectations on
$n$ qubits~\cite{Sciorilli2025}.
%\revcut{In important difference is that PCE accommodates much more general loss functions and decoding than the particular QRAO algorithm.}
\revnext{An important difference is that PCE supports broader loss functions and decoding procedures than the QRAO scheme considered here.}
\revcut{The premise of Theorem~\ref{thm:prep} extends to $k>1$ only with an efficient readout}
\revcut{preserving a centered encoded-to-decoded relation along with end-to-end accounting.}
%At any $k$, the classical gain gap again excludes a uniform polynomial-cost PCE pipeline returning on every input a cut with fixed positive classical gain fraction and inverse-polynomial success, unless $\NP\subseteq\BQP$.

\paragraph*{Implications for PCE.}
More explicitly, consider a PCE family with a uniform polynomial-cost
decoder $D_k$ and a centered encoded quantity
$\widetilde g_k^{\mathrm{enc}}(\rho)$ satisfying
\begin{equation}
 \mathbb E[C_G(D_k(\rho))]-\frac{W}{2}
 \ge \kappa\,\widetilde g_k^{\mathrm{enc}}(\rho),
 \qquad \kappa\ge\kappa_0>0,
 \label{eq:pce-transfer}
\end{equation}
where $\kappa_0$ is independent of the instance size, and the expectation
includes all measurement and decoding randomness.
Then %an all-instance 
a \revcut{worst case}\revnext{worst-case} preparation guarantee
$\widetilde g_k^{\mathrm{enc}}(\rho)\ge\tau\gstar$, for fixed
$\tau>0$ and with inverse-polynomial operational success, would imply
$\NP\subseteq\BQP$ by the same reduction as
Theorem~\ref{thm:prep}.
If the transfer factor vanishes with instance size, or if the optimized
PCE loss is not consistent with the final decoder, the present theorem
does not by itself convert an encoded-state guarantee into a fixed
decoded-gain guarantee.
Independently of such a transfer relation, the classical gain gap excludes,
at any compression order $k$, a uniform polynomial-cost PCE pipeline
that returns on every input a final cut recovering a fixed positive
fraction of the optimal classical gain, unless
$\NP\subseteq\BQP$.

For PCE, this end-to-end accounting includes circuit training, estimation of the
Pauli correlations, sign extraction, repair or local search, and
repetitions. Small observable magnitude margins can therefore move cost from
qubit width into the number of state preparations and measurements,
even when the number of measurement settings is small
\cite{Sciorilli2025,HadfieldNFC2026}.\footnote{Here a measurement
setting means a distinct measurement-basis or POVM configuration,
including any required basis-change circuit.} %; repeated shots within
%one setting are counted separately.} 

%Subsequent QRAO work develops recursive rounding and
%nonvariational preparation heuristics~\cite{Kondo2025,He2025}.

%[new] For PCE, this end-to-end accounting includes estimation of the Pauli
%correlations, sign extraction, repair or local search, and repetitions.
%Small observable margins can therefore move cost from qubit width into
%the number of state preparations and measurements, even when the number
%of measurement settings is small~\cite{Sciorilli2025,HadfieldNFC2026}.

\subsection{Compressed calibration and gain accounting}
The following two-qubit example makes the distinction between the classical threshold value and the compressed quantum optimum concrete. 
%attaining the classical target $\Cstar$ is not the same as
%A two-qubit example calibrates the gap between the classical threshold and the compressed quantum optimum. 
Consider packing three disjoint unit-weight edges on the three Pauli axes to give 
\begin{equation}
 H_3=\tfrac32(\Id-X_1X_2-Y_1Y_2-Z_1Z_2).
 \label{eq:matching-shortest}
\end{equation}
Here $\Cstar=3$ while $\lambda_{\max}(H_3)=6$. 
Bitwise-opposite QRAC codewords are product states of energy $3$, whereas the %singlet has energy $6$. 
two-qubit singlet
$|\psi^-\rangle=(|01\rangle-|10\rangle)/\sqrt{2}$, which is perfectly
anticorrelated along all three Pauli axes, attains the optimum energy $6$. 
Magic rounding gives expected cut values $5/3$ for the product state
and $2$ for the singlet. Thus the singlet's encoded-energy increase
of $3$ over product states raises the expected decoded cut by $1/3$, giving exactly the
$1/d^2=1/9$ transfer for $d=3$.
Although entanglement can improve QRAO performance on some instances~\cite{TeramotoEntanglement2023}, it is unnecessary for reaching the classical threshold.

More generally, combining Eqs.~\eqref{eq:round} and
\eqref{eq:target} gives
\[
 \overline C-\frac W2
 \ge\frac{\tau}{d^2}\left(\Cstar-\frac W2\right),
\]
so %\revcut{recovers a gain fraction $\gamma\ge\tau/d^2$}
\revcut{the decoded mean recovers decoded gain ratio $\Rdec{\rho}\ge\tau/d^2$}\revnext{the decoded output has gain ratio $\Rdec{\rho}\ge\tau/d^2$}. 
At the classical-threshold target $\tau=1$, the premise becomes  $E_\rho\ge\Cstar$. %and Eq.~\eqref{eq:round} gives
%$ \overline C\ge  W/2+(\Cstar-W/2)/d^2$. 
With $c=\Cstar/W$, dividing by $\Cstar$ yields the
instance-dependent \revcut{approximation ratio}\revnext{ordinary approximation-ratio} guarantee
\begin{equation}
 R\ge\frac1{d^2}+\frac{1-1/d^2}{2c},
 \label{eq:conditional-shortest}
\end{equation}
which is close to one when $c=\Cstar/W$ is close to $1/2$. 
Since $1/2<c\le1$, the right-hand side is minimized at $c=1$,
giving the uniform bounds $R\ge5/8$ for $d=2$ and $R\ge5/9$ for
$d=3$.
%
%Eq.~\eqref{eq:round} maps the encoded energy target $\genc{\rho}\ge\tau\gstar$ to decoded gain ratio at least $\tau/d^2$. At the classical-threshold target $\tau=1$, this implies the familiar uniform ordinary approximation-ratio bounds $5/8$ and $5/9$ for $d\in\{2,3\}$, but the instance-dependent ordinary approximation ratio required by the same preparation premise is
%\begin{equation}
% R\ge\frac1{d^2}+\frac{1-1/d^2}{2c},
% \label{eq:conditional-shortest}
%\end{equation}
%which is close to one when $c=\Cstar/W$ is close to $1/2$. 
%Figure~\ref{fig:gain} shows why a near-unity ordinary approximation ratio need not mean that the required gain has been recovered. 
Figure~\ref{fig:gain} shows how the random-cut baseline alone can
produce an \revcut{approximation ratio}\revnext{ordinary approximation ratio} near one, and so a high value need not guarantee recovery of the \revcut{decoded gain fraction $1/d^2$}\revnext{decoded gain ratio $1/d^2$}. 

\par\medskip\noindent\begin{minipage}{\columnwidth}
 \centering
 \begin{tikzpicture}
 \begin{axis}[
   width=0.98\columnwidth,
   height=0.82\columnwidth,
   xmin=0.5,xmax=1.0,
   ymin=0.48,ymax=1.02,
   xlabel={\revlatest{Optimal cut fraction} $c=C^\star/W$},
   ylabel={Approximation ratio $R=\overline C/C^\star$},
   xtick={0.5,0.6,0.7,0.8,0.9,1.0},
   ytick={0.5,0.6,0.7,0.8,0.9,1.0},
   legend style={at={(0.03,0.03)},anchor=south west,draw=none,fill=none,font=\scriptsize},
   tick label style={font=\scriptsize},
   label style={font=\small},
   domain=0.5:1.0,
   samples=120,
   axis line style={thin},
 ]
   \addplot+[no marks,thick] {1};
   \addlegendentry{\revlatest{Ising} ($d=1$)}
   \addplot+[no marks,thick] {0.25+0.375/x};
   \addlegendentry{\revlatest{QRAO} ($d=2$)}
   \addplot+[no marks,thick] {1/9+4/(9*x)};
   \addlegendentry{\revlatest{QRAO} ($d=3$)}
   \addplot+[no marks,thick,dashed] {1/(2*x)};
   \addlegendentry{Random cut}
   \addplot+[no marks,thick,dotted] {0.878567};
   \addlegendentry{GW guarantee}
 \end{axis}
 \end{tikzpicture}
 
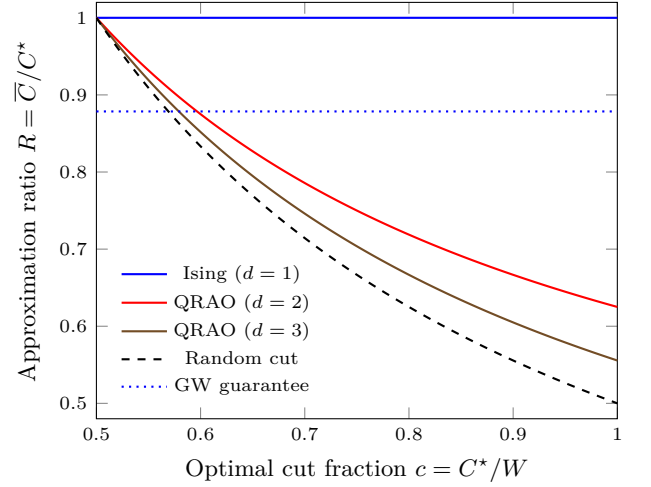
\captionof{figure}{Gain and ordinary approximation ratio for the standard Ising ($d=1$) and compressed QRAO ($d=2,3$) encodings. %\revlatest{Here $W$ is the total edge weight, $\Cstar$ is the optimal cut value, $c=\Cstar/W$, and $R=\overline C/\Cstar$.} 
 For fixed $0<\tau\le1$, the encoded energy condition $E_\rho-W/2\ge\tau(\Cstar-W/2)$ \revcut{related}\revnext{relates} the centered energy of the quantum state \revcut{with}\revnext{to} the optimal classical gain of the problem instance\revcut{,} and yields \revcut{decoded gain fraction at least $\tau/d^2$}\revnext{decoded gain ratio at least $\tau/d^2$}. The solid curves show the resulting conditional bounds on $R$ at the classical-threshold target $\tau=1$, equivalent to $E_\rho\ge\Cstar$. \revlatest{They show consequences of that target, not guarantees achieved by known efficient algorithms. The dashed curve is the uniform classical random cut ratio, and the dotted line is the Goemans-Williamson guarantee.} As $c$ approaches $1/2$, the random guessing \revcut{approximation ratio}\revnext{ordinary approximation ratio} approaches one even though its classical gain is zero.}
 \label{fig:gain}
\end{minipage}\par\medskip

% ===== BEGIN shared/near_tight.tex =====
\subsectionview{near-tight}{Nearly tight compressed relaxations}

Since a decoded cut cannot exceed $\Cstar$, Eq.~\eqref{eq:round} gives
\begin{equation}
 \Cstar\le\lambda_{\max}(H_d)\le W/2+d^2\gstar.
 \label{eq:tight}
\end{equation}
The barrier also applies to guarantees measured relative to the
relaxed quantum optimum. %Since $\lambda_{\max}(H_d)\ge\Cstar$, 
\revcut{Any fixed $0<\beta_{\rm q}\le1$ satisfying}
\revnext{Equivalently, any state satisfying $\Rq{\rho}\ge r$ for a fixed $0<r\le1$, or}
\[
 E_\rho-\frac W2
 \ge\revcut{\beta_{\rm q}}\revnext{r}\left[\lambda_{\max}(H_d)-\frac W2\right]\revnext{,}
\]
also satisfies Eq.~\eqref{eq:target} with
$\tau=\revcut{\beta_{\rm q}}\revnext{r}$.

A small relative relaxation gap need not make the classical preparation threshold accessible.

\begin{revisedproposition}[Near-tight compressed relaxations]
\label{cor:near}
Fix $d\in\{2,3\}$ and a constant $0<\tau\le1$.
Let $\omega=\Id/2^n$ denote the maximally mixed state.
There is a family of fixed packings of exactly $dn$ variables into
$n$ qubits on which any uniform operational preparation procedure
meeting Eq.~\eqref{eq:target} on every input with inverse-polynomial
success would imply $\NP\subseteq\BQP$, even though
\begin{equation}
 \frac{\lambda_{\max}(H_d)}{\Cstar}=1+\Theta(1/n),\;\; %\qquad
 \frac{E_\omega}{\lambda_{\max}(H_d)}=1-\Theta(1/n).
 \label{eq:vanish}
\end{equation}
%\revlatest{The first ratio measures the relative excess of the relaxed
%quantum optimum over the classical optimum, and the second is the energy
%approximation ratio of the maximally mixed state.}
Since $E_\omega=W/2$, we have
$\genc{\omega}=\gdec{\omega}=0$.
The positive constants depend only on $d,\tau$.
\revcut{The separating objective scale is $\Theta(W/n)$, so a robust implementation requires aggregate preparation and readout error $O(W/n)$, equivalently $O(1/n)$ after normalization by $W$.}
\revnext{The separating objective scale is $\Theta(W/n)$. Let $\zeta W$ bound the additive shortfall in encoded energy gain and $\xi W$ bound the additive loss in the decoded mean due to readout. The resulting decoded-mean loss is at most $(\zeta/d^2+\xi)W$, which must be at most $c_0W/n$ for a sufficiently small constant $c_0>0$ to preserve the YES--NO separation. A trace-distance preparation error $T$ is covered by $\zeta\le dT$, as in App.~\ref{app:decoding}.}
\end{revisedproposition}
\begin{proof}
Let $G$ be an $m$-vertex instance of the classical gap problem
\eqref{eq:gap}, with total edge weight $W_0$. Set
$a=\tau/d^2$. Choose the constants $\varepsilon>0$ and $M\ge2$ as
in \revcut{Appendix}\revnext{App.}~\ref{app:padding}, with $M$ independent of $m$. Add to
$G$ a disjoint weighted clique on $2Mm$ vertices whose total edge
weight is $mW_0$. Take $d$ copies of this padded graph and pack
corresponding vertices from the copies onto the same qubit, using
one Pauli axis per copy. The resulting instance has exactly $dn$
variables on $n=(1+2M)m$ 
qubits and total edge weight $W_{\rm tot}=d(m+1)W_0$.

The padding ensures the optimal gain is only an $O(1/m)$ fraction of
$W_{\rm tot}$, while preserving enough of the original YES--NO
separation. Specifically, Eq.~\eqref{eq:paddedgap} shows that,
after normalization by $W_{\rm tot}$, the decoded expectation on a
successful YES preparation exceeds every possible NO-instance cut
by more than
\[
 \frac{a\varepsilon}{2(m+1)}=\Omega(1/m).
\]
The mean-to-tail argument used in the proof of Theorem~\ref{thm:prep} therefore
gives inverse-polynomial probability of crossing a rational
threshold between the two cases. \revcut{Combining this with the assumed inverse-polynomial preparation success probability and polynomially many repetition again decides the NP-hard gap in BQP.} \revnext{Combining the two inverse-polynomial probabilities and repeating the procedure polynomially many times again yields a BQP decision procedure for the NP-hard gap.} The construction uses polynomially
many edges and polynomial-bit rational weights.

It remains to establish the two asymptotic ratios. From
Eq.~\eqref{eq:paddedgain}, the full instance satisfies
$g^\star_{\rm tot}/W_{\rm tot}=O(1/m)$. Equation~\eqref{eq:tight}
then gives
\[
 \frac{\lambda_{\max}(H_d)}{\Cstar}-1=O(1/m),
 \qquad
 1-\frac{E_\omega}{\lambda_{\max}(H_d)}=O(1/m).
\]
For the matching lower bounds, the packed clique Hamiltonian
\eqref{eq:paddingham} has a product of singlets as a maximizing
state. Its quantum optimum exceeds both its classical optimum and
its maximally mixed-state energy by $\Theta(W_0)$, shown in 
Eq.~\eqref{eq:paddingexcess}. The source and padding terms act on
disjoint qubits, while the total classical and quantum optimum
values are each $\Theta(mW_0)$. Both %displayed 
deficits are therefore
$\Omega(1/m)$. Since $n=(1+2M)m$ with fixed $M$, both are
$\Theta(1/n)$ as claimed.
\end{proof}

\revcut{This gives an asymptotic worst-case construction, not one concerning generic instances.}\revnext{This is an asymptotic worst-case construction, not a statement about generic instances.} The fixed padding factor~$M$ inherits %unspecified and 
potentially large constants from the classical gap, so the family is not proposed as a near-term benchmark. Theorem~\ref{thm:prep} tolerates sufficiently small fixed normalized error, whereas the \revcut{corollary}\revnext{proposition} has the inverse-polynomial budget stated above. At $d=1$, the same padding preserves hardness and the mixed-state deficit, but the relaxation is exactly tight.

% ===== END shared/near_tight.tex =====

The bounded-degree construction above also limits this near-tight phenomenon. If the graph has maximum degree~$D$, Eq.~\eqref{eq:bounded-shortest} implies
\begin{equation}
 \frac{E_\omega}{\lambda_{\max}(H_d)}
 \le\frac{W/2}{\Cstar}\le1-\frac1{2D}.
 \label{eq:degree-shortest}
\end{equation}
Hence the energy approximation ratio $E_\omega/\lambda_{\max}(H_d)\ge1-A/n$ for a constant $A>0$, as in \revcut{Corollary~\ref{cor:near}}\revnext{Proposition~\ref{cor:near}}, requires $D\ge n/(2A)$. Linear maximum-degree growth is necessary for this vanishing energy approximation deficit, not merely an artifact of the clique padding of the proof. At the same time this does not say that every hard approximation problem requires large degree.

\section{Outlook}\label{sec:discussion-shortest}
%The preparation barrier already occurs in the standard uncompressed encoding. \revlatest{Theorem~\ref{thm:prep} concerns the cost of selecting a useful state for optimization. Corollary~\ref{cor:near} further elucidates what compression adds. We confirm that relaxation quality need not predict that cost, and a state with zero encoded energy gain and zero decoded mean gain can have energy approximation ratio approaching one.}

\revnext{The quantum state preparation barriers of Theorem~\ref{thm:prep} and  Proposition~\ref{prop:ordinary-ratio} cover both standard and compressed encodings, on near-term devices as well as future fault-tolerant quantum hardware. Proposition~\ref{cor:near} shows that a nearly-tight optimized relaxation need not make the %required 
quantum energy gain efficiently accessible as improvement in the output solution.}
%\revcut{The quantum state preparation barrier from complexity arises in both standard and compressed encodings.}
%\revcut{Theorem~\ref{thm:prep} gives the general result, and Corollary~\ref{cor:near}}
%\revcut{further elucidates what compression adds.}
For QRAO and its recently proposed extensions~\cite{Kondo2025,He2025,SuzukiDecoder2026}, a particular implementation is covered by the theorem only when its complete
preparation-and-decoding pipeline satisfies
Definition~\ref{def:operational}. 
%\revcut{We confirm that relaxation quality need not predict that cost,}
%\revcut{and a state with zero encoded energy gain and zero decoded mean gain}
%\revcut{can still have energy approximation ratio approaching one.}

%Ordinary approximation ratios can give vacuous decoded-gain bounds when $\gstar$ is small, so the theorem is a worst-case exclusion, not a performance ranking. %Comparisons should use common instances, a common random-cut baseline, and matched end-to-end budgets including preparation, optimization, shots, restarts, readout, precision, decoding or repair, and classical postprocessing~\cite{Abbas2024,BernalNeira2024}. 
%

Our results directly support the use of diverse benchmarking criteria in fairly assessing the performance of quantum algorithms, and make clear that reporting any one ratio alone is often insufficient. 
Comparisons should use common instances, gain and classical
baselines along with \revcut{approximation ratio}\revnext{the ordinary approximation ratio}, and matched end-to-end budgets including preparation,
optimization, distinct measurement settings, shots per setting,
restarts, readout, precision, decoding or repair, and classical
postprocessing~\cite{Abbas2024,BernalNeira2024}. 
%
%\revlatest{Following Table~\ref{tab:measures}, report $R$ and the energy approximation ratio as mean-quality diagnostics, the encoded gain ratio $\Renc{\rho}$, the decoded gain ratio $\Rdec{\rho}$, and $p_\alpha$ as target-reaching frequency. Pair them with feasibility, the decoded gain distribution, and time to target under the matched end-to-end budget. When $\Cstar$ is unavailable, report the unnormalized decoded mean gain and cut-value threshold or state the certified bound used for normalization. Circuit depth or any one ratio alone is insufficient.} Such tests neither prove nor refute the theorem, and a low certified decoded gain ratio need not imply poor empirical cuts.
%
%\revcut{In terms of performance (cf. Table~\ref{tab:measures}), our results make clear}
%\revcut{that reporting any one ratio alone is insufficient.}

%Decoder-consistent POVM pullbacks formalize when a known baseline and uniformly efficient decoder give the fixed-factor centered-value transfer needed here, and when this can fail for broader objectives~\cite{SuzukiDecoder2026}. Multi-qubit codes and recursive schemes qualify only after their decoder and end-to-end cost are established. 
%
%Existing QRAO extensions include recursive rounding and
%nonvariational state-preparation heuristics~\cite{Kondo2025,He2025}.
%A particular implementation is covered only when its complete
%preparation-and-decoding pipeline satisfies
%Definition~\ref{def:operational}.
%

The themes of the paper % can be extended 
extend beyond MaxCut. 
Other optimization problems likewise require a sampleable baseline, classical hardness for recovering a fixed fraction of the advantage, and an efficient encoding and readout map. Boolean constraint satisfaction problems are natural generalizations because advantage-over-random and hardness of approximation theories are well developed classically~\cite{Hastad2001,HastadVenkatesh2004,Raghavendra2008}, but each still requires a problem-specific construction to obtain rigorous barriers.

\paragraph*{Open problems.} Four boundaries suggest concrete next steps. (i) Determine the largest uniformly achievable \revcut{decoded classical gain fraction}\revnext{decoded gain ratio} and any corresponding preparation barrier on bounded-degree graph classes. APX-completeness on cubic graphs does not by itself supply the gain gap used here. (ii) For weight-$k>1$ PCEs, identify an efficient decoder with a centered-value identity strong enough to transfer a state-energy premise. (iii) Extend the construction to other constraint satisfaction problems, with their own baselines, gaps, encodings, and readouts. (iv) Establish random-instance barriers from distributional hardness or a worst-to-average-case reduction for a common or practically relevant ensemble.

\par Our results identify worst-case complexity boundaries for quantum optimization algorithms. They leave open the broader question of which problems and instance classes admit a potential quantum advantage. The larger challenge remains to identify regimes in which quantum methods can or cannot outperform the best classical algorithms under comparable end-to-end costs.

%These worst-case barriers assume $\NP\not\subseteq\BQP$. They neither lower-bound entanglement or known-answer preparation depth nor exclude suitable restricted-ansatz parameters. Structured promises and classically attainable solution quality remain separate.

%\revnext{Two concurrent preprints by the author share the QRAO encoding but have distinct outputs. Reference~\cite{HadfieldQRAOComplexity2026} classifies NP-, StoqMA-, and QMA-complete Hamiltonian promise problems, including their behavior under compilation, without analyzing preparation success or end-to-end cost at the classical threshold. Reference~\cite{HadfieldNFC2026} proves information-theoretic tradeoffs among compression, uniformly achievable observable margins, and readout cost, without using the MaxCut-Gain gap.}

%\revnext{Here, Theorem~\ref{thm:prep} combines the classical gain gap with the QRAO decoder to constrain uniform preparation on every instance, and Proposition~\ref{cor:near} constructs near-tight relaxations on which the constraint persists. Neither preprint enters the proofs. The shared objects are the encodings, while decision complexity, information recovery, and operational preparation remain distinct.}

% ===== END shortest/discussion.tex =====

% ===== BEGIN shortest/acknowledgments.tex =====
\begin{acknowledgments}
%\begin{minipage}{0.96\linewidth}
S.H. acknowledges support from U.S. Department of Energy under grant No. DE-SC0026126. %\revcut{OpenAI ChatGPT assisted with literature synthesis, correctness checking of all results, and general polishing of the text. The author has verified correctness of all content and assumes full responsibility for the manuscript.} 
\revnext{OpenAI ChatGPT assisted with literature synthesis, consistency checks, and language editing. The author verified correctness of all content and assumes full responsibility for the manuscript.}
%\end{minipage}
\end{acknowledgments}
\paragraph*{\revnext{Data availability.}} \revnext{No empirical data were created or analyzed. All proofs and results needed to support the conclusions appear in the article.}

% ===== END shortest/acknowledgments.tex =====

\makeatletter\expandafter\gdef\csname b@apsrev42Control\endcsname{0}\makeatother
\nocite{apsrev42Control}
\bibliographystyle{apsrev4-2}
\bibliography{bib}

\appendix

% ===== BEGIN shortest/decoding.tex =====
\techheading{decoding}{Local decoding and error tolerance}
\revnext{This appendix supplies the decoding and error-tolerance details used in Sec.~\ref{sec:model} and the proof of Theorem~\ref{thm:prep}.}
For~$d$ pairwise anticommuting single-qubit Pauli axes $\sigma_a\in\{X,Y,Z\}$, $\sigma_a=Z$ for $d=1$, define a positive operator-valued measure
(POVM) with outcomes $z\in\{-1,1\}^d$ and POVM elements
\begin{equation}
 M_z=2^{-d}\left(\Id+d^{-1/2}\sum_{a=1}^d z_a\sigma_a\right).
 \label{eq:povm-shortest}
\end{equation}
The $M_z$ are each positive semidefinite,  sum to $\Id$, and
outcome $z$ occurs with probability $\Tr(\rho M_z)$.
The expectation over the outcomes of their product measurement satisfies $\E[z_i z_j]=\Tr(\rho P_iP_j)/d$ on every edge, since its endpoints occupy different qubits. This proves Eq.~\eqref{eq:round}, including for entangled $\rho$. 

For %signs $z_1,\ldots,z_d$ assigned to one qubit, 
%the corresponding
%QRAC code state is
%\[
% \frac12\left(\Id+d^{-1/2}\sum_{a=1}^d z_a\sigma_a\right).
%\]
each qubit 
its Bloch-vector component along axis $\sigma_a$ is $z_a/\sqrt d$.
Taking the tensor product of these single-qubit states over all
packed qubits gives the product state QRAC embedding of a classical
assignment.
%The Bloch vector $d^{-1/2}(z_1,\ldots,z_d)$ gives the product QRAC embedding. 
Unused slots may be filled arbitrarily. The local states and measurements admit polynomial-precision implementations.

Let $\zeta,\xi\ge0$ denote preparation and readout error budgets normalized by $W$. If preparation gives encoded energy gain at least $\tau\gstar-\zeta W$ and the readout loss, averaged over the decoder's measurement and classical randomness, is at most $\xi W$, the reduction works whenever
\begin{equation}
 \frac{\tau\varepsilon-\zeta}{d^2}-\xi>\eta_\varepsilon.
 \label{eq:error-shortest}
\end{equation}
Indeed, choose a rational threshold between the NO upper bound and the successful YES mean and apply Eq.~\eqref{eq:tail}. A trace-distance error $T\ge0$ between the intended and prepared states changes energy by at most $dWT$, because $\|H_d-W\Id/2\|\le dW/2$. Thus sufficiently small fixed implementation errors are allowed for Theorem~\ref{thm:prep}. \revcut{Appendix}\revnext{App.}~\ref{app:padding} instead requires inverse-polynomial error.

If runtime, averaged over all internal classical randomness and measurement outcomes, is at most $t(N)$ and success is at least $s(N)$, truncation at $2t(N)/s(N)$ loses at most $s(N)/2$ of the success probability. The expected-runtime extension therefore holds for known uniform polynomial bounds. Postselection, resets, and failed trials are included in this accounting.

% ===== END shortest/decoding.tex =====

\techheading{auxiliary}{Bounded-degree gain construction}
\revnext{This appendix proves the bounded-degree construction stated in Sec.~\ref{sec:uncompressed}, Eq.~\eqref{eq:bounded-shortest}.}
If the maximum degree is at most $D$, greedily choose a heaviest remaining edge and delete its incident edges. This gives a matching $\mathcal M$ with $w(\mathcal M)\ge W/(2D-1)$ because each selected edge removes at most $2D-1$ edges, all no heavier. Orient matched pairs oppositely with independent fair orientations and assign unmatched vertices independently. With respect to these independent classical choices, every edge outside the matching is cut with probability $1/2$. Conditional expectation derandomizes the construction and proves Eq.~\eqref{eq:bounded-shortest}.

% ===== BEGIN shared/technical/padding.tex =====
\techheading{padding}{Vanishing relative gap with hard gain recovery}
\revnext{This appendix proves the padding estimates used in Proposition~\ref{cor:near}.}
Here tildes denote quantities for one padded graph, the subscript ${\rm tot}$ denotes the $d$-copy packed instance, and the subscript ${\rm pad}$ denotes its clique terms. 
Fix $d,\tau$ and set $a=\tau/d^2$. Choose fixed $\varepsilon$ with $\eta_\varepsilon<a\varepsilon/4$, and an integer $M\ge2$ with $1/(4M-2)<a\varepsilon/4$. To an $m$-vertex source graph of weight $W$, add a disjoint clique on $k=2Mm$ vertices \branchadd{(so $k$ is even)}. Give each clique edge weight $u=2mW/[k(k-1)]$, so its total weight is $mW$. Its classical gain divided by its weight is
\begin{equation}
 h_m=\frac{1}{2(2Mm-1)},\qquad mh_m\le\frac1{4M-2}<\frac{a\varepsilon}{4}.
 \label{eq:cliquegain}
\end{equation}
The padded graph has $\widetilde W=(m+1)W$ and
\begin{equation}
 \frac{\widetilde g^\star}{\widetilde W}
 =\frac{\gstar/W+mh_m}{m+1}\le\frac1{m+1}.
 \label{eq:paddedgain}
\end{equation}
Take $d$ copies and pack corresponding vertices. Then $n=(1+2M)m$ and $W_{\rm tot}=d(m+1)W$. Eq.~\eqref{eq:tight} gives $\lambda_{\max}/\Cstar_{\rm tot}-1=O(1/n)$ and $1-E_\omega/\lambda_{\max}=O(1/n)$, with upper bound constant $2d^2(1+2M)$ sufficient for both.

A successful YES preparation has decoded mean gain divided by $W_{\rm tot}$ at least $a(\varepsilon+mh_m)/(m+1)$, whereas every NO cut has classical gain divided by $W_{\rm tot}$ at most $(\eta_\varepsilon+mh_m)/(m+1)$. Their separation is
\begin{equation}
 \frac{a\varepsilon-\eta_\varepsilon-(1-a)mh_m}{m+1}
 >\frac{a\varepsilon}{2(m+1)}.
 \label{eq:paddedgap}
\end{equation}
Choosing the midpoint as threshold gives a successful decoded cut with probability $\Omega(1/m)$ per successful preparation. Thus polynomially many repetitions suffice, also when preparation succeeds only with inverse-polynomial probability. The construction uses $O(m^2)$ edges and polynomial-bit rational weights. Fixed $M$ may be large but is independent of $m$. Normalized error budgets of at most $c_0/m$, for a sufficiently small fixed $c_0>0$, preserve the gap.

For the matching lower bounds in Eq.~\eqref{eq:vanish}, let $S_a=\sum_{j=1}^k\sigma_a^{(j)}$, where $\sigma_a^{(j)}$ is Pauli axis $a$ acting on clique qubit $j$. The packed clique Hamiltonian is exactly
\begin{equation}
 H_{\rm pad}=\left(\frac{dmW}{2}+\frac{d^2uk}{4}\right)\Id
 -\frac{du}{4}\sum_{a=1}^dS_a^2.
 \label{eq:paddingham}
\end{equation}
The final sum is positive semidefinite. A product of $k/2$ singlets is annihilated by every $S_a$, so the displayed scalar is $\lambda_{\rm pad}$. The classical optimum of the $d$ cliques is $\Cstar_{\rm pad}=dmW/2+duk/4$. Hence, for $d\in\{2,3\}$,
\begin{equation}
 \lambda_{\rm pad}-\Cstar_{\rm pad}
 =\frac{d(d-1)uk}{4}=\Theta(W).
 \label{eq:paddingexcess}
\end{equation}
Because source and padding act on disjoint qubits and the source quantum optimum is at least its classical optimum, division by $\Cstar_{\rm tot}=\Theta(mW)$ proves the first lower bound. The padding contributes $d^2uk/4=\Theta(W)$ above the mixed-state baseline, proving the second. With the upper bounds, both rates are $\Theta(1/n)$.

% ===== END shared/technical/padding.tex =====

\end{document}